\documentclass[submission,copyright,creativecommons]{eptcs}
\providecommand{\event}{AFL 2023} 

\usepackage{iftex}
\usepackage{combelow}

\ifpdf
  \usepackage{underscore}         
  \usepackage[T1]{fontenc}        
\else
  \usepackage{breakurl}           
\fi

\usepackage{extarrows,paralist,amssymb,amsmath}
\usepackage{tikz,pgf}\usetikzlibrary{positioning,arrows,automata,backgrounds}

\title{Strictly Locally Testable and Resources Restricted Control Languages in Tree-Controlled Grammars}
\author{Bianca Truthe 
\institute{Institut f\"ur Informatik, Universit\"at Giessen, Arndtstr.~2, 35392 Giessen, Germany}
\email{bianca.truthe@informatik.uni-giessen.de}}
\def\titlerunning{Strictly Locally Testable \& Resources Restricted Control Languages in Tree-Controlled Grammars}
\def\authorrunning{B.~Truthe}

\def\Set#1#2{\left\{\: #1\;|\; #2\:\right\}}

\def\set#1#2{\{\; #1 \mid #2\;\}}
\def\sets#1{\{#1\}}

\newcommand{\CIRC}{\mathit{CIRC}}
\newcommand{\COMB}{\mathit{COMB}}
\newcommand{\COMM}{\mathit{COMM}}
\newcommand{\DEF}{\mathit{DEF}}
\newcommand{\FIN}{\mathit{FIN}}
\newcommand{\MON}{\mathit{MON}}
\newcommand{\NC}{\mathit{NC}}
\newcommand{\NIL}{\mathit{NIL}}
\newcommand{\ORD}{\mathit{ORD}}
\newcommand{\PS}{\mathit{PS}}
\newcommand{\REG}{\mathit{REG}}
\newcommand{\SUF}{\mathit{SUF}}
\newcommand{\UF}{\mathit{UF}}
\newcommand{\RL}{\mathit{RL}}
\newcommand{\SLT}{\mathit{SLT}}
\newcommand{\MAT}{\mathit{MAT}}
\newcommand{\fin}{\mathit{fin}}
\newcommand{\EOL}{\mathit{E0L}}
\newcommand{\ETOL}{\mathit{ET0L}}
\newcommand{\CF}{\mathit{CF}}
\newcommand{\CS}{\mathit{CS}}

\newcommand{\State}{\mathit{State}}
\newcommand{\Var}{\mathit{Var}}
\newcommand{\Prod}{\mathit{Prod}}

\def\cF{{\cal F}}

\def\cTC{{\cal TC}}

\def\slt#1#2#3#4{\textbf{\rm [}#1,#2,#3,#4\textbf{\rm ]}}

\def\Lra{\Longrightarrow}
\def\ra{\rightarrow}

\tikzstyle{to}=[->, >=stealth]
\tikzstyle{hier}=[->, >=angle 60]
\tikzstyle{hiero}=[->, >=angle 60, dashed]
\tikzstyle{state}=[circle,draw,inner sep=2pt,minimum size=8mm]
\tikzstyle{edgeLabel}=[inner sep=0.5mm,fill=white,text=black!60]

\newtheorem{theorem}{Theorem}[section]

\newtheorem{lemma}[theorem]{Lemma}

\newtheorem{example}[theorem]{Example}
\newenvironment{proof}{{\em Proof. }}{{}\hspace*{\fill}$\Box$ \par \medskip }
\newenvironment{proof*}{{\em Proof. }}{\par \medskip }

\newlength{\btlabelwidth}
\newlength{\btleftmargin}

\begin{document}
\maketitle
\begin{abstract}
Tree-controlled grammars are context-free grammars where the derivation process is controlled
in such a way that every word on a level of the derivation tree 
must belong to a certain control language.
We investigate the generative capacity of such tree-controlled grammars where the control 
languages are special regular sets, especially strictly locally testable languages or
languages restricted by resources of the generation (number of non-terminal symbols or
production rules) or acceptance (number of states).
Furthermore, the set theoretic inclusion relations of these subregular language families 
themselves are studied.
\end{abstract}

\section{Introduction}

In the monograph \cite{DasPau89} by J{\"{u}}rgen~Dassow and Gheorghe~P{\u{a}}un, \emph{Seven
Circumstances Where Context-Free Grammars Are Not Enough} are presented. A
possibility to enlarge the generative power of context-free grammars is to introduce
some regulation mechanism which controls the derivation in a context-free grammar.
In some cases, regular languages are used for such a regulation. They are rather easy
to handle and, used as control, they often lead to context-sensitive or even recursively
enumerable languages while the core grammar is only context-free.

One such control mechanism was introduced by Karel~{\v{C}}ulik~II and Hermann~A.~Maurer 
in \cite{CulMau77} where the structure of derivation trees of context-free grammars is restricted by the
requirement that the words of all levels of the derivation tree must
belong to a given regular (control) language. This model is called tree-controlled grammar.

Gheorghe~P{\u{a}}un proved that the generative capacity of such grammars coincides with 
that of context-sensitive grammars (if no erasing rules are used) or
arbitrary phrase structure grammars (if erasing rules are used). Thus, the question arose to 
what extend the restrictions can be weakened in order to obtain `useful' families of languages
which are located somewhere between the classes of context-free and context-sensitive languages.

In \cite{DasStiTru-dcfs08-tcs,DasStiTru-afl08-ijfcs,DasTru-dcfs08,DasTru-afl08,Sti08-collMitrana,TurDasManSel12,Vasz12},
many subregular families of languages have been investigated as classes for the control languages.
In this paper, we continue this research with further subregular language families, especially
strictly locally testable languages or languages restricted by resources of the generation 
(number of non-terminal symbols or production rules) or acceptance (number of states).
Furthermore, the set theoretic inclusion relations of these subregular language families 
themselves are studied.

\section{Preliminaries}
Throughout the paper, we assume that the reader is familiar with the basic concepts of the 
theory of automata and formal languages. 
For details, we refer to~\cite{handbook}. Here we only recall some notation and the definition 
of contextual grammars with selection which form the central notion of the paper.

\subsection{Languages, grammars, automata}

Given an alphabet $V$, we denote by $V^*$ and $V^+$ the set of all words and the set of all non-empty words over $V$,
respectively. The empty word is denoted by~$\lambda$. By $V^k$,
and $V^{{}\leq k}$ for some natural number $k$,
we denote the set of all words of the alphabet $V$ with exactly $k$ letters
and the set of all words over $V$ with at most $k$ letters, respectively.
For a word $w$ and a letter $a$, we denote the length of $w$ by $|w|$ and the number of occurrences of the letter $a$ in the
word~$w$ by $|w|_a$. For a set $A$, we denote its cardinality by $|A|$.

A right-linear grammar is a quadruple $G=(N,T,P,S)$ where $N$ is a finite set of non-terminal symbols,~$T$ is a finite set of
terminal symbols, $P$ is a finite set of production rules of the form $A\ra wB$ or~$A\ra w$ with $A,B\in N$ and~$w\in T^*$,
and $S\in N$ is the start symbol. Such a grammar is called regular, if all the rules are of the form $A\ra xB$
or $A\ra x$ with $A,B\in N$ and $x\in T$ or $S\ra\lambda$. The language generated by a right-linear or regular grammar
is the set of all words over the terminal alphabet which are obtained from the start symbol $S$ by a successive replacement
of the non-terminal symbols according to the rules in the set $P$. 
A non-terminal symbol $A$ is replaced by the right-hand side $w$ of a rule $A\ra w\in P$ in order 
to derive the next sentential form. The language generated consists of all sentential forms 
without a non-terminal symbol.
Every language generated by a right-linear grammar
can also be generated by a regular grammar.

A deterministic finite automaton is a quintuple $A=(V,Z,z_0,F,\delta)$ where $V$ is a finite set of input symbols, $Z$
is a finite set of states, $z_0\in Z$ is the initial state, $F\subseteq Z$ is a set of accepting states, and $\delta$ is
a transition function $\delta: Z\times V\to Z$. The language accepted by such an automaton is the set of all input words 
over the alphabet $V$ which lead letterwise by the transition function from the initial state to an accepting state.

A regular expression over an alphabet $V$ is defined inductively as follows:
\begin{enumerate}
\item $\emptyset$ is a regular expression;
\item every element $x\in V$ is a regular expression;
\item if $R$ and $S$ are regular expressions, so are the concatenation $R\cdot S$,
the union~$R\cup S$, and the Kleene closure $R^*$;
\item for every regular expression, there is a natural number $n$ such that the 
regular expression is obtained from the atomic elements $\emptyset$ and $x\in V$
by $n$ operations concatenation, union, or star.
\end{enumerate}

The language $L(R)$ which is described by a regular expression $R$ is also 
inductively defined:
$L(\emptyset)=\emptyset$; $L(x)=\sets{x}$ for each $x\in V$;
and $L(R\cdot S) = L(R)\cdot L(S)$, $L(R\cup S) = L(R)\cup L(S)$, and $L(R^*) = (L(R))^*$ for
regular expressions $R$ and $S$.

The set of all languages generated by some right-linear grammar coincides with the set of all languages accepted by a
deterministic finite automaton and with the set of all languages described by a regular expression. All these languages 
are called regular and form a family denoted by $\REG$. Any subfamily
of this set is called a subregular language family.

A context-free grammar is a quadruple $G=(N,T,P,S)$ where 
$N$, $T$, and $S$ are as in a right-linear grammar but
the production rules in the set $P$ are of the form $A\ra w$ with~$A\in N$ and~$w\in (N\cup T)^*$.

The language generated by a context-free grammar is the set of all words over the 
terminal alphabet which are obtained from the start symbol $S$ by replacing sequentially the non-terminal symbols 
according to the rules in the set $P$. A language is called context-free if it is generated by some context-free grammar.
The family of all context-free languages is denoted by $\CF$.

With a derivation of a terminal word by a context-free grammar, we associate a derivation tree which has the start symbol
in its root and where every node with a non-terminal $A\in N$ has as children nodes with symbols which form, read from left to
right, a word $w$ such that $A\ra w$ is a rule of the grammar (if~$A\ra\lambda$, then the node with $A$ has only one
child node and this is labelled with $\lambda$). Nodes with terminal symbols or $\lambda$ have no children.
With any derivation tree $t$ of height $k$ and any number~$0\leq j\leq k$, we associate the word of level $j$ 
and the sentential form of level $j$ which are given by all nodes of depth~$j$ read from left to right and all nodes
of depth~$j$ and all leaves of depth less than $j$ read from left to right, respectively.
Obviously, if two words $w$ and $v$ are sentential forms of two successive levels, then $w\Lra^* v$ holds
and this derivation is obtained by a parallel replacement of all non-terminal symbols occurring in the word $w$.

A context-sensitive grammar is a quadruple $G=(N,T,P,S)$ where $N$ is a finite set of non-terminal symbols, 
$S\in N$ is the start symbol, $T$ is a finite set of terminal symbols, and $P$ is a finite set of production 
rules of the form $\alpha\ra\beta$ with $\alpha\in (N\cup T)^+\setminus T^*$, $\beta\in (N\cup T)^*$, 
and $|\beta|\geq|\alpha|$ with the only exception that $S\ra\lambda$ is allowed if the sysmbol $S$ does not occur
on any right-hand side of a rule. The language generated by a context-sensitive grammar is the set of all words over the 
terminal alphabet which are obtained from the start symbol $S$ by replacing sequentially subwords
according to the rules in the set $P$. A language is called context-sensitive if it is generated by some context-sensitive
grammar. The family of all context-sensitive languages is denoted by $\CS$.
For every context-sensitive language $L$, there is a context-sensitive grammar $G=(N,T,P,S)$ with $L(G)=L$, 
where all rules in $P$ are of the form 
\[AB\ra CD,\ A\ra BC,\ A\ra B, \mbox{ or } A\ra a\]
with $A,B,C,D\in N$ and $a\in T$, or $S\ra\lambda$ if $S$ does not occur on the right-hand side of a rule.
Such a grammar is said to be in Kuroda normal form (\cite{Kur64}).

We also mention here four classes of languages without a definition since they are mentioned only in the summary of existing
results: By $\MAT$, we denote the family of all languages generated by 
matrix grammars with appearance checking and without erasing rules; 
by $\MAT_{\fin}$, we denote the family of all such languages where the matrix grammar is of finite index 
(\cite{DasPau89}, \cite{handbook}).
By $\EOL$ ($\ETOL$), we denote the family of all languages generated by
extended (tabled) interactionless Lindenmayer systems (\cite{RS-Lsystems}).

\subsection{Complexity measures and resources restricted languages}

Let $G=(N,T,P,S)$ be a right-linear grammar, $A=(V,Z,z_0,F,\delta)$ be a deterministic finite
automaton, and $L$ be a regular language. Then, we recall the following complexity measures from \cite{DasManTru11b}:
\begin{align*}
\State(A) &= |Z|, \Var(G) = |N|, \Prod(G) = |P|, \\
\State(L)&= \min \Set{\State(A)}{ A \mbox{ is a det. finite automaton accepting }L }, \\
\Var_\RL(L)  &= \min \Set{\Var(G)  }{ G \mbox{ is a right-linear grammar generating }L }, \\
\Prod_\RL(L) &= \min \Set{\Prod(G) }{ G \mbox{ is a right-linear grammar generating }L }.
\end{align*}
We now define subregular families by restricting the resources needed
for generating or accepting their elements:
\begin{align*}
\RL_n^V &= \Set{ L }{ L\in\REG \mbox{ with } \Var_\RL(L)\leq n },\\
\RL_n^P &= \Set{ L }{ L\in\REG \mbox{ with } \Prod_\RL(L)\leq n },\\
\REG_n^Z &= \Set{ L }{ L\in\REG \mbox{ with } \State(L)\leq n }.
\end{align*}

\subsection{Subregular language families based on the structure}

We consider the following restrictions for regular languages. Let $L$ be a language \pagebreak
over an alphabet $V$. 
With respect to the alphabet $V$, the language $L$ is said to be

\begin{itemize}
\item \emph{monoidal} if and only if $L=V^*$,
\item \emph{nilpotent} if and only if it is finite or its complement $V^*\setminus L$ is finite,
\item \emph{combinational} if and only if it has the form
$L=V^*X$
for some subset $X\subseteq V$,
\item \emph{definite} if and only if it can be represented in the form
$L=A\cup V^*B$
where~$A$ and~$B$ are finite subsets of $V^*$,
\item \emph{suffix-closed} (or \emph{fully initial} or \emph{multiple-entry} language) if
and only if, for any two words~$x\in V^*$ and~$y\in V^*$, the relation $xy\in L$ implies
the relation~$y\in L$,
\item \emph{ordered} if and only if the language is accepted by some deterministic finite
automaton
\[A=(V,Z,z_0,F,\delta)\]
with an input alphabet $V$, a finite set $Z$ of states, a start state $z_0\in Z$, a set $F\subseteq Z$ of
accepting states and a transition mapping $\delta$ where $(Z,\preceq )$ is a totally ordered set and, for
any input symbol~$a\in V$, the relation $z\preceq z'$ implies $\delta (z,a)\preceq \delta (z',a)$,
\item \emph{commutative} if and only if it contains with each word also all permutations of this
word,
\item \emph{circular} if and only if it contains with each word also all circular shifts of this
word,
\item \emph{non-counting} (or \emph{star-free}) if and only if there is a natural
number $k\geq 1$ such that, for every three words $x\in V^*$, $y\in V^*$, and $z\in V^*$, it 
holds~$xy^kz\in L$ if and only if $xy^{k+1}z\in L$,
\item \emph{power-separating} if and only if, there is a natural number $m\geq 1$ such that
for every word~$x\in V^*$, either
$J_x^m \cap L = \emptyset$
or
$J_x^m\subseteq L$
where
$J_x^m = \set{ x^n}{n\geq m}$,
\item \emph{union-free} if and only if $L$ can be described by a regular expression which
is only built by product and star,
\item \emph{strictly locally $k$-testable} if and only if there are three subsets $B$, $I$, and $E$ of $V^k$
such that any word~$a_1a_2\ldots a_n$ with $n\geq k$ and $a_i\in V$ for $1\leq i\leq n$ belongs to the language $L$
if and only if
\begin{gather*}
a_1a_2\ldots a_k\in B,\\
a_{j+1}a_{j+2}\ldots a_{j+k}\in I \text{ for every $j$ with $1\leq j\leq n-k-1$ and}\\
a_{n-k+1}a_{n-k+2}\ldots a_n\in E,
\end{gather*}
\item \emph{strictly locally testable} if and only if it is strictly locally $k$-testable for some natural number $k$.
\end{itemize}

We remark that monoidal, nilpotent, combinational, definite, ordered, union-free, and strictly locally~($k$-)testable
languages are regular, whereas non-regular languages of the other types mentioned above exist.
Here, we consider among the commutative, circular, suffix-closed, non-counting,
and power-separating languages only those which are also regular.

Some properties of the languages of the classes mentioned above can be found in
\cite{Shyr91} (monoids),
\cite{GecsegPeak72} (nilpotent languages),
\cite{Ha69} (combinational and commutative languages),
\cite{PerRabSham63} (definite languages),
\cite{GilKou74} and \cite{BrzoJirZou14} (suffix-closed languages),
\cite{ShyThi74-ORD} (ordered languages),
\cite{Das79} (circular languages),
\cite{McNPap71} (non-counting and strictly locally testable languages),
\cite{ShyThi74-PS} (power-separating languages),
\cite{Brzo62} (union-free languages).

By $\FIN$, $\MON$, $\NIL$, $\COMB$, $\DEF$, $\SUF$, $\ORD$, $\COMM$, $\CIRC$, $\NC$, $\PS$, $\UF$, 
$\SLT_k$ (for any natural number $k\geq 1$), and $\SLT$, 
we denote the families of all finite, 
monoidal, nilpotent, combinational, definite, regular suffix-closed, ordered, regular commutative, regular circular, 
regular non-counting, regular power-separating, union-free, 
strictly locally $k$-testable, and strictly locally testable languages, respectively.

For any natural number $n\geq 1$, let $\MON_n$ be the set of all languages that
can be represented in the form $A_1^*\cup A_2^*\cup\cdots\cup A_k^*$
with $1\leq k\leq n$ where all $A_i$ ($1\leq i\leq k$) are alphabets. Obviously,
\[\MON=\MON_1\subset\MON_2\subset\dots\subset\MON_j\subset \cdots.\]

A strictly locally testable language characterized by three finite sets $B$, $I$, and~$E$ as above which includes
additionally a finite set $F$ of words which are shorter than those of the sets $B$, $I$, and~$E$
is denoted by $\slt{B}{I}{E}{F}$.

As the set of all families under consideration, we set
\begin{align*}
\mathfrak{F} &= \sets{ \FIN, \NIL, \COMB, \DEF, \SUF, \ORD, \COMM, \CIRC, \NC, \PS, \UF}\\
    &\quad{}\cup\set{\MON_k}{k\geq 1}\cup\sets{\SLT}\cup\set{\SLT_k}{k\geq 1}\\
    &\quad{}\cup\set{ \RL_n^V}{n\geq 1}\cup\set{\RL_n^P}{n\geq 1}\cup\set{\REG_n^Z}{n\geq 1}.
\end{align*}

\subsection{Hierarchy of subregular families of languages}

In this section, we present a hierarchy of the families of the aforementioned set $\mathfrak{F}$ with respect to the
set theoretic inclusion relation. A summary is depicted in Figure~\ref{fig-subreg-hier}.

Before this, we prove some relations of the classes of strictly locally $k$-testable languages 
to the subregular language families restricted by resources, which have not been 
considered in the literature yet. 

For this purpose, we first introduce some languages which serve later as witness languages for proper
inclusions and incomparabilities.

\begin{lemma}\label{l-l1}
The language $L_1=\{a\}^*\{b\}\{a,b\}^*$ belongs to $\REG_2^Z\setminus\SLT$.
\end{lemma}
\begin{proof}
The language $L_1$ is accepted by the automaton with two states whose transition function is 
given in the following diagram (double-circled states are accepting):

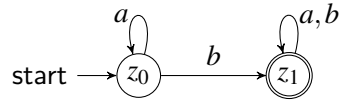
\begin{center}
\begin{tikzpicture}[on grid,>=stealth',initial text={\sf start}
]
\node[state,minimum size=4mm,initial] (z_0) at (0,0) {$z_0$};
\node[state,minimum size=4mm,accepting] (z_1) at (2,0) {$z_1$};
\draw[->] (z_0) edge node [above,sloped] {$b$} (z_1);
\draw[->] (z_0) .. controls +(75:1) and +(105:1) .. node [left] {$a$} (z_0);
\draw[->] (z_1) .. controls +(75:1) and +(105:1) .. node [right] {$a,b$} (z_1);
\end{tikzpicture}
\end{center}

Suppose, the language $L_1$ is strictly locally $k$-testable for some natural number $k\geq 1$.
Then, there exist sets $B\subseteq V^k$, $I\subseteq V^k$, $E\subseteq V^k$, and $F\subseteq V^{{}\leq k-1}$
such that $L_1=\slt{B}{I}{E}{F}$.
Since the word $a^{2k}ba^{2k}$ belongs to the language $L_1$, we know that $a^k\in B\cap I\cap E$. But
then, also the word $a^{2k}$ belongs to the language which is a contradiction.
\end{proof}

\begin{lemma}\label{l-l2}
The language $L_2=\slt{\{a,b\}}{\{b,c\}}{\{a,c\}}{\emptyset}$ belongs to $\SLT_1\setminus\REG_4^Z$.
\end{lemma}
\begin{proof}
By definition, $L_2\in\SLT_1$.

We now prove that $L_2$ is not accepted by an deterministic finite automaton with less than five states.
Let $L=L_2$ and let $R_L$ be the Myhill-Nerode equivalence relation (see \cite{HopUll79}): 
two words $x$ and $y$ are in this relation if and only if,
for all words $z$, either both words $xz$ and $yz$ belong to the language $L$ or none of them.
The words $\lambda$, $a$, $b$, $c$, and $aa$ are pairwise not in this relation, as one can check.

Therefore, the index of the language $L$ is at least five. Hence, at least five states are necessary for
accepting the language $L$.
\end{proof}

\begin{lemma}\label{l-l3}
For each natural number $n\geq 2$, let $V_n=\{a_1,a_2,\ldots,a_{n-1}\}$ be an alphabet with $n-1$ pairwise
different letters and let $L_{3,n}=\{a_1a_2\ldots a_{n-1}\}$.
Then, every language $L_{3,n}$ for $n\geq 2$ belongs to the set~$\SLT_2\setminus\REG_n^Z$.
\end{lemma}
\begin{proof}
The statement $L_{3,n}\in\SLT_2$ for $n\geq 2$ can be seen as follows.
If $n=2$, then $L_{3,n}=\slt{\emptyset}{\emptyset}{\emptyset}{\{a_1\}}$, otherwise 
$L_{3,n}=\slt{\{a_1a_2\}}{\set{a_pa_{p+1}}{2\leq p\leq n-3}}{\{a_{n-2}a_{n-1}\}}{\emptyset}$.

For accepting any language $L_{3,n}$ for $n\geq 2$, at least $n+1$ states are necessary (follows from the
fact that the $n$ partial words $a_1\ldots a_i$ for $0\leq i\leq n-1$ and $a_1a_1$ are pairwise not in the 
Myhill-Nerode relation).
\end{proof}

\begin{lemma}\label{l-l4}
For each natural number $n\geq 1$, let $L_{4,n}=\{a^n\}$.
Then $L_{4,n}$ belongs to the set $\RL_1^P\setminus\SLT_n$.
\end{lemma}
\begin{proof}
The single word $a^n$ can be generated with one rule, hence, $L_{4,n}\in\RL_1^P$.

Assume that such a language is strictly locally $n$-testable. Then, it is $L_{4,n}=\slt{B}{I}{E}{F}$ for
suitable sets $B$, $I$, $E$, and $F$. 
From $L_{4,n}=\{a^n\}$, it follows that $B=E=\{a^n\}$.
But then, also the word $a^{n+1}$ belongs to the language $L_{4,n}$ which is a contradiction.
\end{proof}

\begin{lemma}\label{l-l5}
For each natural number $n\geq 1$, let $V_n=\{a_1,a_2,\ldots,a_n\}$ be an alphabet with $n$ pairwise
different letters and let $L_{5,n}=V_n^*$.
Then, for $n\geq 1$, the language $L_n$ belongs to the set $\SLT_1\setminus\RL_n^P$.
\end{lemma}
\begin{proof}
The language $L_{5,n}$ can be represented as
$L_{5,n}=\slt{V}{V}{V}{\{\lambda\}}$. Hence, $L_{5,n}\in\SLT_1$ for $n\geq 1$.

For generating a language $L_{5,n}$ for some number $n\geq 1$, at least a non-terminating rule is necessary for every 
letter $a_i$ ($1\leq i\leq n$) and additionally a terminating rule. Hence, $L_{5,n}\notin\RL_n^P$.
\end{proof}

\begin{lemma}\label{l-l6}
The language $L_6=\{a\}$ belongs to $\RL_1^V\setminus\SLT_1$.
\end{lemma}
\begin{proof}
The language $L_6$ can be generated with a single rule and, hence, with one non-terminal only.

Assume that $L_6$ is strictly locally 1-testable and can be represented as $\slt{B}{I}{E}{F}$. Then
$B=E=\{a\}$. But then, also the word $aa$ belongs to the language which is a contradiction.
\end{proof}

\begin{lemma}\label{l-l7}
The language $L_7=\{a\}\{b\}^*\{a\}\cup\{a\}$ belongs to $\SLT_1\setminus\RL_1^V$.
\end{lemma}
\begin{proof}
The language $L_7$ is strictly locally 1-testable and can be represented as $\slt{\{a\}}{\{b\}}{\{a\}}{\emptyset}$.

Assume that the language $L_7$ is generated by a right-linear grammar with one non-terminal symbol only.
Let $m$ be the maximal length of the right-hand side of a rule: $m=\max(\set{w}{S\ra w\in P})$. Then, the
word $ab^ma$ cannot be derived in one step. Hence, there is a derivation $S\Lra ab^pS\Lra^* ab^ma$ for some number $p$
with $0\leq p\leq m-2$. But then, also the derivation $S\Lra ab^pS\Lra ab^pab^pS\Lra^* ab^pab^ma$ is possible
which yields a word which does not belong to the language $L_7$. Due to this contradiction, 
we obtain that $L_7\notin\RL_1^V$.
\end{proof}

\begin{lemma}\label{l-l8}
The language $L_8=\set{a^{3m}}{m\geq 1}$ belongs to $\RL_1^V\setminus\SLT$.
\end{lemma}
\begin{proof}
The language $L_8$ is generated by the right-linear grammar $G=(\{S\},\{a\},\{S\ra a^3S,\ S\ra a^3\},S)$.
Hence, $L_8\in\RL_1^V$.

Assume that the language $L_8$ is generated by a strictly locally $k$-testable grammar for some number~$k\geq 1$.
Then, $L_8$ has a representation as $\slt{B}{I}{E}{F}$ with $B\cup I\cup\ E\subseteq\{a\}^k$ and $F\subseteq \{a\}^{{}\leq k-1}$. 
Since the word~$a^{3k}$ belongs to the language $L_8$, we obtain that $B$, $I$, and $E$ contain the word $a^k$. 
But then, also the word $a^{3k+1}$  
belongs to the language $L_8$ which is a contradiction.
\end{proof}

\begin{lemma}\label{l-l9}
For each natural number $n\geq 1$, let $V_n=\{a_1,a_2,\ldots,a_{n+1}\}$ be an alphabet with $n+1$ pairwise
different letters and let $L_{9,n}=\{a_1\}^+\{a_2\}^+\cdots\{a_{n+1}\}^+$.
Then, for $n\geq 1$, the language $L_{9,n}$ belongs to the set $\SLT_2\setminus\RL_n^V$.
\end{lemma}
\begin{proof}
The language $L_{9,n}$ can be represented as
\[L_{9,n}=\slt{\{a_1a_1,a_1a_2\}}{\set{a_pa_p}{1\leq p\leq n+1}\cup\set{a_pa_{p+1}}{1\leq p\leq n}}{\{a_na_{n+1},a_{n+1}a_{n+1}\}}{\emptyset}.\]
Hence, $L_{9,n}\in\SLT_2$ for $n\geq 1$.

For generating a language $L_{9,n}$ for some number $n\geq 1$, at least a non-terminal symbol is necessary for every 
letter $a_i$ ($1\leq i\leq n+1$). Hence, $L_{9,n}\notin\RL_n^V$.
\end{proof}

We now prove inclusion relations and incomparabilities.

\begin{lemma}\label{lem-slt1-reg5}
The class $\SLT_1$ is properly included in the class $\REG_5^Z$.
\end{lemma}
\begin{proof}
We first prove the inclusion $\SLT_1\subseteq\REG_5^Z$.

Let $L$ be a strictly locally 1-testable language.
Then $L=\slt{B}{I}{E}{F}$ with $B\subseteq V$, $I\subseteq V$, $E\subseteq V$, and~$F\subseteq \sets{\lambda}$. 
We construct the following deterministic finite automaton:
\[A=(V,\sets{z_0,z_1,\ldots,z_4},z_0,Z_{\mathrm{f}},\delta)\]
where 
\[Z_{\mathrm{f}}=\sets{z_1,z_2}\cup
\begin{cases}
\sets{z_0}, & \text{if $\lambda\in F$},\\
\emptyset,  & \text{otherwise},
\end{cases}
\]
and the transition function $\delta$ is given by the following diagram 
($z_0$ is an accepting state if and only if~$\lambda\in F$):
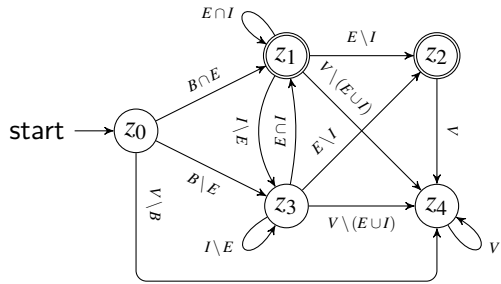
\vspace{-3mm}
\begin{center}
\begin{tikzpicture}[on grid,>=stealth',initial text={\sf start}
]
\node[state,minimum size=4mm,initial] (z_0) at (2,1) {$z_0$};
\node[state,minimum size=4mm,accepting] (z_1) at (4,2) {$z_1$};
\node[state,minimum size=4mm,accepting] (z_2) at (6,2) {$z_2$};
\node[state,minimum size=4mm] (z_3) at (4,0) {$z_3$};
\node[state,minimum size=4mm] (z_4) at (6,0) {$z_4$};
%
\draw[->] (z_0) edge node [above,sloped] {\tiny$B\cap E$} (z_1);
\draw[->] (z_0) edge node [below,sloped] {\tiny$B\setminus E$} (z_3);
\draw[->,rounded corners] (z_0) |- +(1,-2) node [above,sloped,pos=.2] {\tiny$V\setminus B$} -| (z_4);
\draw[->] (z_1) edge [bend right=30] node [below,sloped] {\tiny$I\setminus E$} (z_3);
\draw[->] (z_1) edge node [above,sloped] {\tiny$E\setminus I$} (z_2);
\draw[->] (z_1) edge node [above,sloped,pos=.25] {\tiny$V\setminus(E\cup I)$} (z_4);
\draw[->] (z_1) .. controls +(120:1) and +(150:1) .. node [left] {\tiny$E\cap I$} (z_1);
\draw[->] (z_2) edge node [above,sloped] {\tiny$V$} (z_4);
\draw[->] (z_3) edge [bend right=10] node [above,sloped] {\tiny$E\cap I$} (z_1);
\draw[->] (z_3) edge node [above,sloped,pos=.3] {\tiny$E\setminus I$} (z_2);
\draw[->] (z_3) edge node [below,sloped] {\tiny$V\setminus(E\cup I)$} (z_4);
\draw[->] (z_3) .. controls +(210:1) and +(240:1) .. node [left] {\tiny$I\setminus E$} (z_3);
\draw[->] (z_4) .. controls +(300:1) and +(330:1) .. node [right] {\tiny$V$} (z_4);
\end{tikzpicture}
\end{center}

Due to space reasons, we leave the proof that $L(A)=L$ to the reader.
From the construction follows the inclusion $\SLT_1\subseteq\REG_5^Z$.

A witness language for the properness of this inclusion is the language $L_1=\{a\}^*\{b\}\{a,b\}^*$ from Lemma~\ref{l-l1}.
\end{proof}

\begin{lemma}\label{lem-slt1-reg234}
The class $\SLT_1$ is incomparable to the classes $\REG_i^Z$ for $i\in\sets{2,3,4}$.
\end{lemma}
\begin{proof}
Due to the inclusion relations, it suffices to show that there is a language in the set $\REG_2^Z\setminus\SLT_1$
and a language in the set $\SLT_1\setminus\REG_4^Z$. A language for the first case is $L_1=\{a\}^*\{b\}\{a,b\}^*$
as shown in Lemma~\ref{l-l1}. A language for the second case is $L_2=\slt{\{a,b\}}{\{b,c\}}{\{a,c\}}{\emptyset}$
as shown in Lemma~\ref{l-l2}.
\end{proof}

\begin{lemma}\label{lem-sltk-regn}
The classes $\SLT_k$ for $k\geq 2$ and $\SLT$ are incomparable to the classes $\REG_n^Z$ for $n\geq 2$.
\end{lemma}
\begin{proof}
Due to the inclusion relations, it suffices to show that there is a language in the set $\REG_2^Z\setminus\SLT$
and a language in each set $\SLT_2\setminus\REG_n^Z$ for $n\geq 2$. A language for the first case is $L_1=\{a\}^*\{b\}\{a,b\}^*$
as shown in Lemma~\ref{l-l1}. Languages for the second case are 
$L_{3,n}=\{a_1a_2\ldots a_{n-1}\}$ as shown in Lemma~\ref{l-l3}.
\end{proof}

\begin{lemma}\label{lem-sltk-rlpn}
The classes $\SLT_k$ for $k\geq 1$ are incomparable to the classes $\RL_n^P$ for $n\geq 1$.
\end{lemma}
\begin{proof}
Due to the inclusion relations, it suffices to show that there is a language in the set $\RL_1^P\setminus\SLT_k$ for
every $k\geq 1$ and a language in each set $\SLT_1\setminus\RL_n^P$ for $n\geq 1$. Languages for the first 
case are $L_{4,k}=\{a^k\}$ for $k\geq 1$ as shown in Lemma~\ref{l-l4}. Languages for the second case are 
$L_{5,n}=\{a_1,a_2,\ldots, a_n\}^*$ as shown in Lemma~\ref{l-l5}.
\end{proof}

\begin{lemma}\label{lem-slt1-rlv2}
The class $\SLT_1$ is properly included in the class $\RL_2^V$.
\end{lemma}
\begin{proof}
Let $L=\slt{B}{I}{E}{F}$ be a strictly locally 1-testable language over an alphabet $T$. 
We construct a right-linear grammar $G=(\{S,S'\},T,P,S)$ with the rules 
\begin{itemize}
\item $S\ra w$ for every word $w\in F\cup (B\cap E)$,
\item $S\ra wS'$ for every word $w\in B$,
\item $S'\ra wS'$ for every word $w\in I$, and
\item $S'\ra w$ for every word $w\in E$.
\end{itemize}
The language $L(G)$ generated is $F\cup(B\cap E)\cup (BI^*E)$ which is $L$. Hence, $L\in\RL_2^V$
and $\SLT_1\subseteq\RL_2^V$. A witness language for the properness of the inclusion is $L_6=\{a\}$ for which was proved in
Lemma~\ref{l-l6} that it belongs to the set $\RL_1^V$ and therefore also to $\RL_2^V$ but not to $\SLT_1$.
\end{proof}

\begin{lemma}\label{lem-slt1-rlv1}
The class $\SLT_1$ is incomparable to the class $\RL_1^V$.
\end{lemma}
\begin{proof}
There is a language in the set $\RL_1^V\setminus\SLT_1$, namely $L_6=\{a\}$
as shown in Lemma~\ref{l-l6}, and a language in the set $\SLT_1\setminus\RL_1^V$, 
namely $L_7=\{a\}\{b\}^*\{a\}\cup\{a\}$ as shown in Lemma~\ref{l-l7}.
\end{proof}

\begin{lemma}\label{lem-sltk-rlvn}
The classes $\SLT_k$ for $k\geq 2$ and $\SLT$ are incomparable to the classes $\RL_n^V$ for $n\geq 1$.
\end{lemma}
\begin{proof}
Due to the inclusion relations, it suffices to show that there is a language in the 
set $\RL_1^V\setminus\SLT$ and a language in the set $\SLT_2\setminus\RL_n^V$ for every number $n\geq 1$. 
A language for the first case is $L_8=\set{a^{3m}}{m\geq 1}$ as shown in Lemma~\ref{l-l8}. 
A language for the second case is $L_{9,n}=\{a_1\}^+\{a_2\}^+\cdots\{a_{n+1}\}^+$
as shown in Lemma~\ref{l-l9}.
\end{proof}

A summary of the inclusion relations is given in Figure~\ref{fig-subreg-hier}.
An edge label in this figure refers to the paper or lemma above 
where the respective inclusion is proved. 

\begin{figure}[htb]
\centerline{%
\scalebox{0.75}{\begin{tikzpicture}[node distance=15mm and 25mm,on grid=true,
background rectangle/.style=
{
draw=black!80,
rounded corners=1ex},
show background rectangle]
\node (REG) {$\REG$};
\node (dummy4) [below left=of REG] {};
\node (PS) [left=of dummy4] {$\PS$};
\node (NC) [below=of PS] {$\NC$};
\node (dummy1) [below=of NC] {};
\node (ORD) [below=of dummy1] {$\ORD$};
\node (dummy2) [below=of ORD] {};
\node (DEF) [below=of dummy2] {$\DEF$};
\node (COMB) [below right=of DEF] {$\COMB$};
\node (NIL) [below=of DEF] {$\NIL$};
\node (FIN) [below right=of NIL] {$\FIN$};
\node (SUF) [below left=of ORD] {$\SUF$};
\node (COMM) [below left=of SUF] {$\COMM$};
\node (dummy3) [above=of COMM] {};
\node (CIRC) [above=of dummy3] {$\CIRC$};
\node (SLT1) [right=of DEF] {$\SLT_1$};
\node (SLT2) [above=of SLT1] {$\SLT_2$};
\node (SLTn) [above=of SLT2] {$\vdots$};
\node (SLT) [above=of SLTn] {$\SLT$};
\node (V1) [right=of SLT1] {$\RL_1^V$};
\node (V2) [above=of V1] {$\RL_2^V$};
\node (Vn) [right=of SLTn] {$\vdots$};
\node (Z2) [right=of V1] {$\REG_2^Z$};
\node (Z34) [right=of Vn] {$\vdots$};
\node (Z5) [above=of Z34] {$\REG_5^Z$};
\node (Zn) [above=of Z5] {$\vdots$};
\node (dummy) [right=of Z2] {};
\node (d1) [below=of dummy] {};
\node (d2) [below=of d1] {};
\node (P1) [below left=of d2] {$\RL_1^P$};
\node (P2) [above=of P1] {$\RL_2^P$};
\node (P3) [above right=of P2] {$\RL_3^P$};
\node (P4) [above=of P3] {$\RL_4^P$};
\node (Pn) [right=of Z34] {$\vdots$};
\node (UF) [right=of Pn] {$\UF$};
\node (d3) [below=of FIN] {};
\node (Z1) [below right=of d3] {$\REG_1^Z$};
\node (MON) [below=of Z1] {$\MON$};
\draw[hier] (FIN) to node[edgeLabel]{\footnotesize{\cite{Wi78}}} 
  (NIL);
\draw[hier] (MON) to node[edgeLabel]{\footnotesize{\cite{Tru18-TRsubreg}}} 
  (Z1);
\draw[hier] (Z1) [bend right=-15] to node[edgeLabel]{\footnotesize{\cite{Tru18-TRsubreg}}} 
  (NIL);
\draw[hier] (Z1) [bend right=-29] to node[edgeLabel]{\footnotesize{\cite{Tru18-TRsubreg}}} 
  (SUF);
\draw[hier] (Z1) [bend right=-29] to node[edgeLabel]{\footnotesize{\cite{Tru18-TRsubreg}}} 
  (COMM);
\draw[hier] (Z1) [bend right=39] to node[edgeLabel,pos=.3]{\footnotesize{\cite{Tru18-TRsubreg}}} 
  (UF);
\draw[hier] (Z1) [bend right=0] to node[edgeLabel]{\footnotesize{\cite{Tru18-TRsubreg}}} 
  (Z2);
\draw[hier] (Z1) [bend right=10] to node[edgeLabel,pos=.2]{\footnotesize{\cite{DasTru22-ncma}}} 
  (SLT1);
\draw[hier] (P1) to node[edgeLabel,pos=.5]{\footnotesize{\cite{Tru18-TRsubreg}}} 
  (FIN);
\draw[hier] (P1) [bend right=30] to node[edgeLabel,pos=.3]{\footnotesize{\cite{Tru18-TRsubreg}}} 
  (UF);
\draw[hier] (V1) to node[edgeLabel,pos=.4]{\footnotesize{\cite{Tru18-TRsubreg}}} 
  (V2);
\draw[hier] (V2) to node[edgeLabel,pos=.4]{\footnotesize{\cite{Tru18-TRsubreg}}} 
  (Vn);
\draw[hier] (Vn) to node[edgeLabel]{\footnotesize{\cite{Tru18-TRsubreg}}} 
  (REG);
\draw[hier] (Z2) to node[edgeLabel,pos=.4]{\footnotesize{\cite{Tru18-TRsubreg}}} 
  (V2);
\draw[hier] (Z2) to node[edgeLabel,pos=.4]{\footnotesize{\cite{Tru18-TRsubreg}}} 
  (Z34);
\draw[hier] (Z34) to 
  (Z5);
\draw[hier] (Z5) to 
  (Zn);
\draw[hier] (Zn) [bend right=15] to node[edgeLabel]{\footnotesize{\cite{Tru18-TRsubreg}}} 
  (REG);
\draw[hier] (P1) to node[edgeLabel]{\footnotesize{\cite{Tru18-TRsubreg}}} 
  (P2);
\draw[hier] (P2) to node[edgeLabel,pos=.4]{\footnotesize{\cite{Tru18-TRsubreg}}} 
  (P3);
\draw[hier] (P2) 
  to node[edgeLabel]{\footnotesize{\cite{Tru18-TRsubreg}}} 
  (V1);
\draw[hier] (P3) to node[edgeLabel,pos=.4]{\footnotesize{\cite{Tru18-TRsubreg}}} 
  (P4);
\draw[hier] (P4) to node[edgeLabel,pos=.4]{\footnotesize{\cite{Tru18-TRsubreg}}} 
  (Pn);
\draw[hier] (P4) 
  to node[edgeLabel,pos=.2]{\footnotesize{\cite{Tru18-TRsubreg}}} 
  (V2);
\draw[hier] (Pn) [bend right=22] to node[edgeLabel]{\footnotesize{\cite{Tru18-TRsubreg}}} 
  (REG);
\draw[hier] (NIL) to node[edgeLabel,pos=.4]{\footnotesize{\cite{Wi78}}} 
  (DEF);
\draw[hier] (NIL) [bend left=-6] 
  to node[edgeLabel,pos=.15]{\footnotesize{\cite{Tru18-TRsubreg}}} 
  (V1);
\draw[hier] (COMB) 
  to node[edgeLabel,pos=.7]{\footnotesize{\cite{Ha69}}} 
  (DEF);
\draw[hier] (COMB) to node[edgeLabel,pos=.4]{\footnotesize{\cite{Tru18-TRsubreg}}} 
  (V1);
\draw[hier] (COMB) [bend right=4] to node[edgeLabel,pos=.3]{\footnotesize{\cite{Tru18-TRsubreg}}} 
  (Z2);
\draw[hier] (COMB) to node[edgeLabel,pos=.6]{\footnotesize{\cite{DasTru22-ncma}}} 
  (SLT1);
\draw[hier] (SLT1) to node[edgeLabel,pos=.6]{\footnotesize{\cite{SCR-PSP-11}}} 
  (SLT2);
\draw[hier] (SLT1) to node[edgeLabel,pos=.4]{\footnotesize{\cite{DasTru22-ncma}}} 
  (ORD);
\draw[hier] (SLT1) [bend left=15] to node[edgeLabel,pos=.3]{\footnotesize{\ref{lem-slt1-reg5}}} 
  (Z5);
\draw[hier] (SLT1) [bend right=10] to node[edgeLabel,pos=.3]{\footnotesize{\ref{lem-slt1-rlv2}}} 
  (V2);
\draw[hier] (SLT2) to node[edgeLabel]{\footnotesize{\cite{SCR-PSP-11}}} 
  (SLTn);
\draw[hier] (SLTn) to node[edgeLabel]{\footnotesize{\cite{SCR-PSP-11}}} 
  (SLT);
\draw[hier] (SLT) to node[edgeLabel]{\footnotesize{\cite{McNPap71}}} 
  (NC);
\draw[hier] (ORD) to node[edgeLabel,pos=.4]{\footnotesize{\cite{ShyThi74-ORD}}} 
  (NC);
\draw[hier] (DEF) to node[edgeLabel,pos=.4]{\footnotesize{\cite{HolTru15-ncma}}} 
  (ORD);
\draw[hier] (DEF) to node[edgeLabel,pos=.7]{\footnotesize{\cite{DasTru22-ncma}}} 
  (SLT);
\draw[hier] (DEF) [bend left=-6] to node[edgeLabel,pos=.2]{\footnotesize{\cite{Tru18-TRsubreg}}} 
  (V2);
\draw[hier] (NC) to node[edgeLabel,pos=.4]{\footnotesize{\cite{ShyThi74-PS}}} 
  (PS);
\draw[hier] (PS) [bend right=-8] to node[edgeLabel]{\footnotesize{\cite{HolTru15-ncma}}} 
  (REG);
\draw[hier] (SUF) [bend right=-22] to node[edgeLabel]{\footnotesize{\cite{HolTru15-ncma}}} 
  (PS);
\draw[hier] (COMM) to node[edgeLabel,pos=.4]{\footnotesize{\cite{HolTru15-ncma}}} 
  (CIRC);
\draw[hier] (CIRC) [bend right=-29] to node[edgeLabel]{\footnotesize{\cite{HolTru15-ncma}}} 
  (REG);
\draw[hier] (UF) [bend right=29] to node[edgeLabel]{\footnotesize{\cite{HolTru15-ncma}}} 
  (REG);
\end{tikzpicture}}}
\caption{Hierarchy of subregular language families}\label{fig-subreg-hier}
\end{figure}

\begin{theorem}\label{th-subreg-hier}
The inclusion relations presented in Figure~\ref{fig-subreg-hier} hold.
An arrow from an entry~$X$ to an entry~$Y$ depicts the proper inclusion $X\subset Y$;
if two families are not connected by a directed path, then they are incomparable.
\end{theorem}


\subsection{Tree-controlled grammars}

A tree-controlled grammar is a quintuple $G=(N,T,P,S,R)$ where
\begin{itemize}
\item $(N,T,P,S)$ is a context-free grammar with a set $N$ of non-terminal symbols, a set $T$ of terminal symbols,
a set $P$ of context-free non-erasing rules (with the only exception that the rule $S\ra\lambda$ is allowed if $S$ does not
occur on a right-hand side of a rule), and an axiom $S$,
\item $R$ is a regular set over $N\cup T$.
\end{itemize}

The language $L(G)$ generated by a tree-controlled grammar $G=(N,T,P,S,R)$ consists of all such words~$z\in T^*$ 
which have a derivation tree $t$ where $z$ is the word obtained by reading the leaves from left to right
and the words of all levels of $t$ -- besides the last one -- belong to the regular control language $R$.

Let $\cF$ be a subfamily of $\REG$. Then, we denote the family of languages
generated by tree-controlled grammars $G=(N,T,P,S,R)$ with $R\in \cF$ by ${\cTC}(\cF)$.

\begin{example}\label{ex1}
As an example, we consider the tree-controlled grammar
\[G_1=(\{ S\}, \{a\}, \{S\ra SS, S\ra a\}, S, \{S\}^*).\]
Since the terminal symbol $a$ is not allowed to appear before the last level, on all levels before, any
occurrence of $S$ is replaced by $SS$. Finally, any letter $S$ is replaced by $a$. 
Therefore, the levels of an allowed derivation tree consist of the words
$S$, $SS$, $SSSS$, $\dots$, $S^{2^n}$, $a^{2^n}$ for some $n\geq 0$.
Thus, $L(G_1)=\set{a^{2^n}}{n\geq 0}$.
Due to the structure of the control language which is monoidal and can be generated by a grammar 
with one non-terminal symbol and two rules, we further obtain
\[L(G_1)\in\cTC(\MON)\cap\cTC(\RL_1^V)\cap\cTC(\RL_2^P).\]
\end{example}

\begin{example}\label{ex2}
We now consider the tree-controlled grammar
\[G_2=(\{ S,A,B,C\}, \{ a,b,c\}, P, S, \{ S, aAbBcC\})\]
with
\[P=\{ S\ra aAbBcC, A\ra aA, B\ra bB, C\ra cC, A\ra a, B\ra b, C\ra c\}.\]
By the definition of the control language, any derivation in $G_2$ has the form
\[S\Lra aAbBcC\Lra aaAbbBccC \Lra \dots \Lra a^{n-1}Ab^{n-1}Bc^{n-1}C \Lra a^nb^nc^n\]
with $n\geq 2$. Thus, the tree-controlled grammar $G_2$ generates the non-context-free language
\[L(G_2)=\sets{ a^nb^nc^n}{n\geq 2}.\]
Due to the structure of the control language which is finite and can be generated by a grammar 
with one non-terminal symbol and two rules, we further obtain
\[L(G_2)\in\cTC(\FIN)\cap\cTC(\RL_1^V)\cap\cTC(\RL_2^P).\]
\end{example}

%

In \cite{Pau79-tc} (see also \cite{DasPau89}), it has been shown that a language $L$ is 
generated by a tree-controlled grammar if and only if it is generated by a 
context-sensitive grammar. 
\begin{theorem}\label{th1}{\rm (\cite{Pau79-tc}, \cite{DasPau89})}
It holds ${\cTC}(\REG)=\CS$. 
\end{theorem}

In subsequent papers, tree-controlled grammars have been investigated where the control language
belongs to some subfamily of the class $\REG$ 
(\cite{DasStiTru-dcfs08-tcs,DasStiTru-afl08-ijfcs,DasTru-dcfs08,DasTru-afl08,Sti08-collMitrana,TurDasManSel12,Vasz12}).
In this paper, we continue this research with further subregular language families.

From the definition follows that the subset relation is preserved under the use of tree-controlled
grammars: if we allow more, we do not obtain less.

\begin{lemma}\label{l-tc-gramm-monoton}
For any two language classes $X$ and $Y$ with $X\subseteq Y$,
we have the inclusion
\[\cTC(X)\subseteq\cTC(Y).\]
\end{lemma}

A summary of the inclusion relations known so far is given in Figure~\ref{fig-tc-hier-old}.
An arrow from an entry~$X$ to an entry~$Y$ depicts the inclusion $X\subseteq Y$;
a solid arrow means proper inclusion; a dashed arrow indicates that it is not known whether the inclusion is proper.
If two families are not connected by a directed path, then they are not necessarily incomparable.
An edge label in this figure refers to the paper where the respective inclusion is proved. 

\begin{figure}[htb]
\centerline{%
\scalebox{0.85}{\begin{tikzpicture}[node distance=18mm and 25mm,on grid=true,
background rectangle/.style=
{
draw=black!80,
rounded corners=1ex},
show background rectangle]
\node (CF) {$\CF$};
\node (Z1) [above=of CF] {$\EOL\stackrel{\mbox{\footnotesize{\cite{DasTru-afl08}}}}{=}\cTC(\MON_1)\stackrel{\mbox{\footnotesize{\cite{DasStiTru-dcfs08-tcs}}}}{=}\cTC(\REG_1^Z)$};
\node (d2) [above right=of Z1] {};
\node (COMB) [right=of d2] {$\cTC(\COMB)$};
\node (FIN) [above=of Z1] {$\cTC(\FIN)\stackrel{\mbox{\footnotesize{\cite{DasTru-afl08}}}}{=}\MAT_{\fin}$};
\node (NIL) [above left=of FIN] {$\cTC(\NIL)$};
\node (DEF) [above=of NIL] {$\cTC(\DEF)$};
\node (Z2) [above=of COMB] {$\cTC(\REG_2^Z)$};
\node (d1) [left=of NIL] {};
\node (MON2) [left=of d1] {$\cTC(\MON_{{}\geq 2})\stackrel{\mbox{\footnotesize{\cite{DasStiTru-dcfs08-tcs}}}}{=}\ETOL$};
\node (COMM) [above=of MON2] {$\cTC(\COMM)\stackrel{\mbox{\footnotesize{\cite{DasTru-afl08}}}}{=}\MAT$};
\node (Z4) [above=of Z2] {$\cTC(\REG_4^Z)$};
\node (CS) [above right=of DEF] 
  {\parbox{9.5cm}{$\CS\stackrel{\mbox{\footnotesize{\cite{Pau79-tc}}}}{=}\cTC(\REG)\stackrel{\mbox{\footnotesize{\cite{DasTru-afl08}}}}{=}\cTC(\CIRC)\stackrel{\mbox{\footnotesize{\cite{DasTru-afl08}}}}{=}\cTC(\SUF)$\\
   $\phantom{\CS}\stackrel{\mbox{\footnotesize{\cite{DasTru-afl08}}}}{=}\cTC(\ORD)\stackrel{\mbox{\footnotesize{\cite{DasTru-afl08}}}}{=}\cTC(\NC)\stackrel{\mbox{\footnotesize{\cite{DasTru-afl08}}}}{=}\cTC(\PS)\stackrel{\mbox{\footnotesize{\cite{DasStiTru-dcfs08-tcs}}}}{=}\cTC(\REG_{{}\geq 5}^Z)$
  }};
\draw[hier] (CF) 
  to node[edgeLabel,pos=.4]{\footnotesize{\cite{RS-Lsystems}}} 
  (Z1);
\draw[hier] (Z1) 
  to node[edgeLabel]{\footnotesize{\cite{DasStiTru-afl08-ijfcs}}} 
  (COMB);
\draw[hier] (FIN) 
  to node[edgeLabel]{\footnotesize{\cite{DasTru-afl08}}} 
  (NIL);
\draw[hier] (Z1) [bend left=29] 
  to node[edgeLabel,pos=.3]{\footnotesize{\cite{DasTru-afl08}}} 
  (NIL);
\draw[hiero] (NIL) 
  to 
  (DEF);
\draw[hiero] (COMB) [bend left=15] 
  to 
  (DEF);
\draw[hiero] (COMB) 
  to 
  (Z2);
\draw[hier] (FIN) 
  to node[edgeLabel]{\footnotesize{\cite{DasPau89}}} 
  (MON2);
\draw[hier] (Z1) [bend left=15] 
  to node[edgeLabel]{\footnotesize{\cite{RS-Lsystems}}} 
  (MON2);
\draw[hier] (MON2) 
  to node[edgeLabel,pos=.4]{\footnotesize{\cite{DasPau89}}} 
  (COMM);
\draw[hiero] (Z2) 
  to 
  (Z4);
\draw[hier] (MON2) 
  to node[edgeLabel,pos=.6]{\footnotesize{\cite{DasStiTru-dcfs08-tcs}}} 
  (Z4);
\draw[hier] (COMM) 
  to node[edgeLabel,pos=.4]{\footnotesize{\cite{DasPau89}}} 
  (CS);
\draw[hiero] (DEF) 
  to 
  (CS);
\draw[hiero] (Z4) 
  to 
  (CS);
\end{tikzpicture}}}
\caption{Hierarchy of subregularly tree-controlled language families}\label{fig-tc-hier-old}
\end{figure}
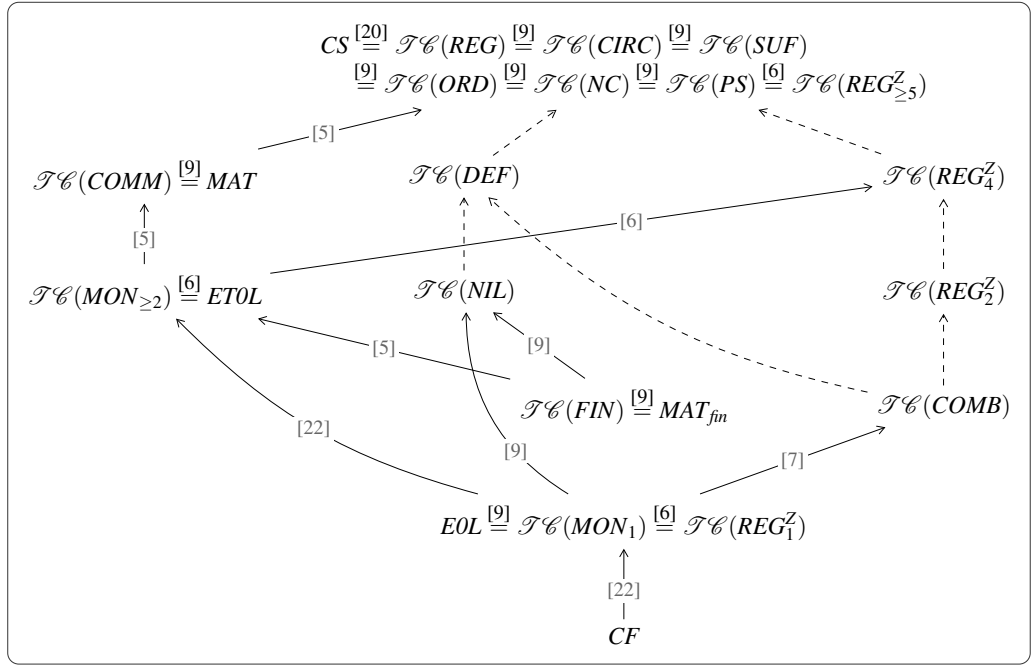

\section{Results}


We insert the classes $\cTC(SLT_k)$ for $k\geq 1$, $\cTC(\SLT)$, $\cTC(\RL_n^V)$ for $n\geq 1$, 
and $\cTC(\RL_n^P)$ for $n\geq 1$ into the existing hierachy (see Figure~\ref{fig-tc-hier-old}).

The inclusions follow from the inclusion relations of the respective families of the control languages
(see Figure~\ref{fig-subreg-hier} and Lemma~\ref{l-tc-gramm-monoton}).

In most cases, we obtain that any context-sensitive language can be generated by a tree-controlled grammar
where the control language is taken from that family.

\begin{theorem}\label{thm-slt2-cs}
We have $\cTC(\SLT_k)=\CS$ for $k\geq 2$ and $\cTC(\SLT)=\CS$.
\end{theorem}
\begin{proof}
Let $L$ be a context-sensitive language. Then, there is a context-sensitive grammar $G=(N,T,P,S)$ with $L(G)=L$
which is in Kuroda normal form, where the rule set $P$ can be divided into two sets~$P_1$ and~$P_2$ such that all rules of $P_1$
are of the form $A\ra BC$ or $A\ra B$ or $A\ra a$ with $A,B,C,D\in N$ and $a\in T$
and all rules of $P_2$ are of the form $AB\ra CD$ with $A,B,C,D\in N$.

We will construct a tree-controlled grammar $G_{\mathrm{tc}}$ which simulates the grammar $G$. Since $G_{\mathrm{tc}}$
has only context-free rules, the non-context-free rules of $G$ have to be substituted by context-free rules and some 
control such that the parts of a non-context-free rule which are independent from the view of the core grammar 
of $G_{\mathrm{tc}}$ remain connected.

We label the non-context-free rules and associate the non-terminal symbols of their left-hand sides with new non-terminal
symbols which are marked with the rule label and the position (first or second letter). The context-free rules can be
freely applied also in the tree-controlled grammar. A non-context-free rule $p:AB\ra CD$ will be simulated by 
context-free rules 
\[A\ra A_{p,1},\ B\ra B_{p,2},\ A_{p,1}\ra C, \mbox{ and } B_{p,2}\ra D.\]
The control language ensures
that the rules which belong together (here $A\ra A_{p,1}$ and $B\ra B_{p,2}$) are applied together (at the same time and 
next to each other). If a terminal symbol is produced in a sentential form of the grammar $G$, then it remains there
until the whole terminal word is produced. In the tree-controlled grammar $G_{\mathrm{tc}}$, one has to keep track
of terminal symbols because they `disappear' (once produced, they are not present in the next level anymore) and then
two non-terminal symbols appear next to each other, although they are not neighbours in the sentential form. So, the 
tree-controlled grammar should produce placeholders for terminal symbols and replace them by the actual terminal symbols
only in the very end. In a tree-controlled grammar, from one level to the next, all non-terminal symbols are
replaced. This can be seen as some kind of shortcut where production rules which are independent from each other are 
applied in parallel.

We construct such a tree-controlled grammar $G_{\mathrm{tc}}=(N_{\mathrm{tc}}, T, P_{\mathrm{tc}}, S, R_{\mathrm{tc}})$.
The terminal alphabet and start symbol are the same as in the grammar $G$. We now give the rules; the non-terminal
symbols will be collected later from the rules. At the end, we will give the control language $R_{\mathrm{tc}}$.

In order to simulate the context-free rules directly, we take all non-terminating rules of them from $G$ as they are:
\[ P_{\mathrm{cf}}=P\cap(\set{A\ra BC}{A,B,C\in N}\cup\set{A\ra B}{A,B\in N}).\]

Instead of the terminating rules, we take rules with a placeholder (for each terminal symbol $a$, we introduce a unique
non-terminal symbol $\hat{a}$), but finally, those placeholders have to be terminated:
\[ P_{\mathrm{t}}=\set{A\ra \hat{a}}{A\in N, a\in T, A\ra a\in P}\cup\set{\hat{a}\ra a}{a\in T}.\]

We give also rules which can delay the derivation such that not everything needs to be replaced in parallel:
\[ P_{\mathrm{d}}=\set{A\ra A}{A\in N}\cup\set{\hat{a}\ra \hat{a}}{a\in T}.\]

For simulating the non-context-free rules, first rules are applied which mark the position of the intended application
such that the control language has the chance to check whether the plan is alright (if it is not, then the derivation will
block). In the next step, the markers will be replaced by their actual target non-terminal symbols:
\[ P_{\mathrm{cs}}=\bigcup_{p:AB\ra CD\in P}\{A\ra A_{p,1},\ B\ra B_{p,2},\ A_{p,1}\ra C,\ B_{p,2}\ra D \}.\]

Other rules are not needed, hence,
\[ P_{\mathrm{tc}}= P_{\mathrm{cf}}\cup P_{\mathrm{t}}\cup P_{\mathrm{d}}\cup P_{\mathrm{cs}}.\]

The set $N_{\mathrm{tc}}$ of non-terminal symbols results as follows:
\begin{gather*}
N_{\mathrm{cf}} = N \cup \set{\hat{a}}{a\in T},\ N_1 = \set{A_{p,1}}{p:AB\ra CD\in P},\ N_2 = \set{B_{p,2}}{p:AB\ra CD\in P},\\
N_{12} = \set{A_{p,1}B_{p,2}}{p:AB\ra CD\in P},\ N_{\mathrm{tc}} = N_{\mathrm{cf}}\cup N_1\cup N_2.
\end{gather*}

A derivation can go wrong only if the simulation of a non-context-free rule is not properly planned. Hence, as control
language, we take
\[ R_{\mathrm{tc}}= (N_{\mathrm{cf}}\cup N_{12})^* .\]

Since the context-free rules of the grammar $G$ can be applied independently from each other and do not have to
be applied at a certain time (thanks to the rules from the subset $P_{\mathrm{d}}$) and the correct simulation 
of the non-context-free rules is ensured by the control language $R_{\mathrm{tc}}$, it is not hard to see
that the generated languages $L(G)$ and $L(G_{\mathrm{tc}})$ coincide.

The control language $R_{\mathrm{tc}}$ is strictly locally 2-testable as can be seen from the following
representation: Let
\begin{align*}
B &= N_{\mathrm{cf}}^2 \cup N_{\mathrm{cf}}N_1 \cup N_{12},&
I &= N_{\mathrm{cf}}^2 \cup N_{\mathrm{cf}}N_1 \cup N_{12} \cup N_2N_{\mathrm{cf}} \cup N_2N_1,\\
E &= N_{\mathrm{cf}}^2 \cup N_{12} \cup N_2N_{\mathrm{cf}},&
F &= N_{\mathrm{cf}}\cup\{\lambda\}.
\end{align*}
Then $R_{\mathrm{tc}}= \slt{B}{I}{E}{F}$.

Altogether, we obtain $\CS\subseteq\cTC(\SLT_2)\subseteq\cTC(\SLT_k)\subseteq\cTC(\SLT)\subseteq\CS$
for $k\geq 3$.
Thus, it holds~$\cTC(\SLT_k)=\CS$ for $k\geq 2$ and $\cTC(\SLT)=\CS$.
\end{proof}

\begin{theorem}\label{thm-rlv-cs}
We have $\cTC(\RL_n^V)=\CS$ for $n\geq 1$.
\end{theorem}
\begin{proof}
The control language $R_{\mathrm{tc}}= (N_{\mathrm{cf}}\cup N_{12})^*$ from the tree-controlled grammar $G_{\mathrm{tc}}$
in the proof of Theorem~\ref{thm-slt2-cs} can be generated by a right-linear 
grammar $G'=(\{S'\},N_{\mathrm{tc}},P',S')$ where 
\[P'=\set{S'\ra xS'}{x\in N_{\mathrm{cf}}\cup N_{12}}\cup\set{S'\ra x}{x\in N_{\mathrm{cf}}\cup N_{12}}.\]
Hence,
$\CS\subseteq\cTC(\RL_1^V)\subseteq\cTC(\RL_n^V)\subseteq\CS$
for $n\geq 2$.
Thus, we conclude $\cTC(\RL_n^V)=\CS$ for $n\geq 1$.
\end{proof}

From the proof of Theorem~\ref{thm-slt2-cs}, we conclude also the following statement.

\begin{theorem}\label{thm-uf-cs}
We have $\cTC(\UF)=\CS$.
\end{theorem}
\begin{proof}
Let $L=\{w_1,w_2,\ldots, w_n\}$ be a finite language. Then $L^*=(\{w_1\}^*\{w_2\}^*\cdots\{w_n\}^*)^*$
and is therefore union-free.

The control language $R_{\mathrm{tc}}= (N_{\mathrm{cf}}\cup N_{12})^*$ from the tree-controlled grammar $G_{\mathrm{tc}}$
in the proof of Theorem~\ref{thm-slt2-cs} is the Kleene closure of a finite language and, hence, it is union-free.
\end{proof}

Regarding the classes $\cTC(\RL_n^P)$ for $n\geq 1$, the situation is different since the
number of rules depends on the size of the alphabet (which is not necessarily the case
for the number of non-terminal symbols or the number of states).

If the control language is generated with one rule only, then either the control language is
the empty set (if the right-hand side of the rule contains a non-terminal symbol) or it contains
exactly one terminal word. Since the start symbol of the tree-controlled grammar always forms the 
first level of the derivation tree, it must be contained in the control language (otherwise, the
derivation would be blocked right from the beginning). Therefore, we obtain the following result.

\begin{lemma}\label{l-p1}
Let $G=(N,T,P,S,R)$ a tree-controlled grammar with $R\in\RL_1^P$. Then, the generated language is
\[L(G)=\begin{cases}\set{w}{w\in T^* \mbox{ and } S\ra w\in P}, & \text{if $R=\{S\}$},\\
  \emptyset, & \text{otherwise}.
\end{cases}\]
\end{lemma}
\begin{proof}
If $R=\{S\}$, then every level but the last one of the derivation tree is $S$ and the last level is
a terminal word which is produced by $S$. On the other hand, all terminal words derived from $S$ belong
to the generated language.

If $R\not=\{S\}$, then $S\notin R$ since $R$ contains at most one word because $R\in\RL_1^P$. Since $S$
is the word of the first level of the derivation tree, there is no derivation possible. Hence, $L(G)$ is empty.
\end{proof}

From this result, the next one immediately follows.

\begin{theorem}\label{thm-p1}
We have $\cTC(\RL_1^P)=\FIN$.
\end{theorem}
\begin{proof}
The inclusion $\cTC(\RL_1^P)\subseteq\FIN$ follows from Lemma~\ref{l-p1}.
The inclusion $\FIN\subseteq\cTC(\RL_1^P)$ can also be seen from Lemma~\ref{l-p1}: Let $L$ be
a finite language over an alphabet $T$. Then, construct a tree-controlled 
grammar $G=(\{S\},T,\set{S\ra w}{w\in L},S,\{S\})$. It holds $L(G)=L$ and $L(G)\in\cTC(\RL_1^P)$.
\end{proof}

If the control language is taken from the family $\cTC(\RL_2^P)$, then already context-sensitive
languages can be generated as the Examples~\ref{ex1} and~\ref{ex2} show.

\begin{theorem}\label{thm-p1-p2}
We have $\cTC(\RL_1^P)\subset\cTC(\RL_2^P)$.
\end{theorem}
\begin{proof}
The inclusion follows from Theorem~\ref{th-subreg-hier} and Lemma~\ref{l-tc-gramm-monoton}.
According to Theorem~\ref{thm-p1}, the family $\cTC(\RL_1^P)$ contains finite languages only. As
shown in the Examples~\ref{ex1} and~\ref{ex2}, the family~$\cTC(\RL_2^P)$ contains non-context-free languages.
\end{proof}

A summary of all the inclusion relations is given in Figure~\ref{fig-tc-hier-new}.
An arrow from an entry~$X$ to an entry~$Y$ depicts the inclusion $X\subseteq Y$;
a solid arrow means proper inclusion; a dashed arrow indicates that it is not known whether the inclusion is proper.
If two families are not connected by a directed path, then they are not necessarily incomparable.
An edge label in this figure refers to the paper or theorem above where the respective inclusion is proved. 

\begin{figure}[htb]
\centerline{%
\scalebox{0.85}{\begin{tikzpicture}[node distance=18mm and 25mm,on grid=true,
background rectangle/.style=
{
draw=black!80,
rounded corners=1ex},
show background rectangle]
\node (P1) {$\cTC(\RL_1^P)\stackrel{\mbox{\footnotesize{\ref{thm-p1}}}}{=}\FIN$};
\node (CF) [above=of P1] {$\CF$};
\node (Z1) [above=of CF] {$\EOL\stackrel{\mbox{\footnotesize{\cite{DasTru-afl08}}}}{=}\cTC(\MON_1)\stackrel{\mbox{\footnotesize{\cite{DasStiTru-dcfs08-tcs}}}}{=}\cTC(\REG_1^Z)$};
\node (d2) [above right=of Z1] {};
\node (COMB) [right=of d2] {$\cTC(\COMB)$};
\node (FIN) [above=of Z1] {$\cTC(\FIN)\stackrel{\mbox{\footnotesize{\cite{DasTru-afl08}}}}{=}\MAT_{\fin}$};
\node (NIL) [above left=of FIN] {$\cTC(\NIL)$};
\node (DEF) [above=of NIL] {$\cTC(\DEF)$};
\node (SLT1) [above left=of COMB] {$\cTC(\SLT_1)$};
\node (Z2) [above=of COMB] {$\cTC(\REG_2^Z)$};
\node (d1) [left=of NIL] {};
\node (MON2) [left=of d1] {$\cTC(\MON_{{}\geq 2})\stackrel{\mbox{\footnotesize{\cite{DasStiTru-dcfs08-tcs}}}}{=}\ETOL$};
\node (COMM) [above=of MON2] {$\cTC(\COMM)\stackrel{\mbox{\footnotesize{\cite{DasTru-afl08}}}}{=}\MAT$};
\node (Z4) [above=of Z2] {$\cTC(\REG_4^Z)$};
\node (P2) [right=of COMB] {$\cTC(\RL_2^P)$};
\node (Pn) [above=of P2] {$\cTC(\RL_n^P)$};
\node (CS) [above right=of DEF] 
  {\parbox{9.7cm}{$\CS\stackrel{\mbox{\footnotesize{\cite{Pau79-tc}}}}{=}\cTC(\REG)\stackrel{\mbox{\footnotesize{\cite{DasTru-afl08}}}}{=}\cTC(\CIRC)\stackrel{\mbox{\footnotesize{\cite{DasTru-afl08}}}}{=}\cTC(\SUF)$\\
   $\phantom{\CS}\stackrel{\mbox{\footnotesize{\cite{DasTru-afl08}}}}{=}\cTC(\ORD)\stackrel{\mbox{\footnotesize{\cite{DasTru-afl08}}}}{=}\cTC(\NC)\stackrel{\mbox{\footnotesize{\cite{DasTru-afl08}}}}{=}\cTC(\PS)\stackrel{\mbox{\footnotesize{\cite{DasStiTru-dcfs08-tcs}}}}{=}\cTC(\REG_{{}\geq 5}^Z)$\\
   $\phantom{\CS}\stackrel{\mbox{\footnotesize{\ref{thm-uf-cs}}}}{=}\cTC(\UF)\stackrel{\mbox{\footnotesize{\ref{thm-slt2-cs}}}}{=}\cTC(\SLT_{{}\geq 2})\stackrel{\mbox{\footnotesize{\ref{thm-slt2-cs}}}}{=}\cTC(\SLT)\stackrel{\mbox{\footnotesize{\ref{thm-rlv-cs}}}}{=}\cTC(\RL_{{}\geq 1}^V)$}};
\draw[hier] (CF) 
  to node[edgeLabel,pos=.4]{\footnotesize{\cite{RS-Lsystems}}} 
  (Z1);
\draw[hier] (Z1) 
  to node[edgeLabel,pos=.7]{\footnotesize{\cite{DasStiTru-afl08-ijfcs}}} 
  (COMB);
\draw[hier] (P1) 
  to node[edgeLabel,pos=.4]{\footnotesize{\cite{handbook}}} 
  (CF);
\draw[hier] (P1) [bend right=29] 
  to node[edgeLabel]{\footnotesize{\ref{thm-p1-p2}}} 
  (P2);
\draw[hiero] (P2) 
  to 
  (Pn);
\draw[hier] (FIN) 
  to node[edgeLabel]{\footnotesize{\cite{DasTru-afl08}}} 
  (NIL);
\draw[hier] (Z1) [bend left=29] 
  to node[edgeLabel,pos=.3]{\footnotesize{\cite{DasTru-afl08}}} 
  (NIL);
\draw[hiero] (NIL) 
  to 
  (DEF);
\draw[hiero] (COMB) [bend left=15] 
  to 
  (DEF);
\draw[hiero] (COMB) 
  to 
  (SLT1);
\draw[hiero] (COMB) 
  to 
  (Z2);
\draw[hier] (FIN) 
  to node[edgeLabel]{\footnotesize{\cite{DasPau89}}} 
  (MON2);
\draw[hier] (Z1) [bend left=15] 
  to node[edgeLabel]{\footnotesize{\cite{RS-Lsystems}}} 
  (MON2);
\draw[hier] (MON2) 
  to node[edgeLabel,pos=.4]{\footnotesize{\cite{DasPau89}}} 
  (COMM);
\draw[hiero] (Z2) 
  to 
  (Z4);
\draw[hier] (MON2) 
  to node[edgeLabel,pos=.6]{\footnotesize{\cite{DasStiTru-dcfs08-tcs}}} 
  (Z4);
\draw[hier] (COMM) 
  to node[edgeLabel,pos=.4]{\footnotesize{\cite{DasPau89}}} 
  (CS);
\draw[hiero] (DEF) 
  to 
  (CS);
\draw[hiero] (SLT1) 
  to 
  (CS);
\draw[hiero] (Z4) 
  to 
  (CS);
\draw[hiero] (Pn) [bend right=25] 
  to 
  (CS);
\end{tikzpicture}}}
\caption{New Hierarchy of subregularly tree-controlled language families}\label{fig-tc-hier-new}
\end{figure}

\section{Conclusion}

There are several families of languages generated by tree-controlled grammars where we do not have 
a characterization by some other language class. The strictness of some inclusions and the 
incomparability of some families remain as open problems.

In the present paper, we have only considered tree-controlled grammars without erasing rules. For 
tree-controlled grammars where erasing rules are allowed, several results have been published already
(see, e.\,g., \cite{DasStiTru-afl08-ijfcs, TurDasManSel12, Vasz12}). Also in this situation, 
there are some open problems.

Another direction for future research is to consider other subregular language families or to
relate the families of languages generated by tree-controlled grammars to language
families obtained by other grammars/systems with regulated rewriting.

\bibliographystyle{eptcs}
\bibliography{tc-subreg}

\end{document}